\documentclass[conference]{IEEEtran}
\IEEEoverridecommandlockouts
\usepackage{cite}
\usepackage{amsmath,amssymb,amsfonts}
\usepackage{algorithmic}
\usepackage{graphicx}
\usepackage{textcomp}
\usepackage{xcolor}
\usepackage{amsthm}

\usepackage{tikz} 
\usetikzlibrary{positioning} 
\usetikzlibrary{calc} 

\newtheorem{theorem}{Theorem}[section]
\newtheorem{lemma}[theorem]{Lemma}

\def\BibTeX{{\rm B\kern-.05em{\sc i\kern-.025em b}\kern-.08em
    T\kern-.1667em\lower.7ex\hbox{E}\kern-.125emX}}
\begin{document}

\title{Near-Field Beamfocusing with Polarized Antennas\\
\thanks{This work is supported by the grant from the Spanish ministry of economic affairs and digital transformation and of the European union – NextGenerationEU UNICO-5G I+D/AROMA3D-Earth (TSI-063000-2021-69), by Grant 2021 SGR 00772 funded by the Universities and Research Department from Generalitat de Catalunya, and by the Spanish Government through
the project 6G AI-native Air Interface (6G-AINA, PID2021-128373OB-I00
funded by MCIN/AEI/10.13039/501100011033) and by ”ERDF A way of
making Europe}
}
\author{\IEEEauthorblockN{Adrian Agustin, Xavier Mestre }
\IEEEauthorblockA{\textit{Information and Signal Processing for Intelligent Communications (ISPIC) Research Unit }\\
\textit{Centre Tecnològic de Telecomunicacions de Catalunya (CTTC/iCERCA)}\\
\texttt{\{adrian.agustin, xavier.mestre\}@cttc.cat}
}
}

\maketitle

\begin{abstract}
One of the most relevant challenges in future 6G wireless networks is how to support a massive spatial multiplexing of a large number of user terminals. Recently, extremely large antenna arrays (ELAAs), also referred to as extra-large MIMO (XL-MIMO), have emerged as an potential enabler of this type of spatially multiplexed transmission. These massive configurations substantially increase the number of available spatial degrees of freedom (transmission modes) while also enabling to spatially focus the transmitted energy into a very small region, thanks to the properties of \textit{near-field} propagation and the large number of transmitters.  This work explores whether multiplexing of multiple orthogonal polarizations can enhance the system performance in the \textit{near-field}. We concentrate on a simple scenario consisting of a Uniform Linear Array (ULA) and a single antenna element user equipment (UE). We demonstrate that the number of spatial degrees of freedom can be as large as 3 in the \textit{near-field} of a Line of Sight (LoS) channel when both transmitter and receiver employ three orthogonal linear polarizations. In the \textit{far-field}, however, the maximum number of spatial degrees of freedom tends to be only 2, due to the fact that the equivalent MIMO channel becomes rank deficient. We provide an analytical approximation to the achievable rate, which allows us to derive approximations to the optimal antenna spacing and array size that maximize the achievable rate.
\end{abstract}


\section{Introduction}
Conventional wireless communication designs have traditionally assumed that radio waves are received in the \textit{far-field} region, where the propagating wave front is essentially a plane. However, a massive deployment of spatially distributed antennas challenges this regime, mainly because the user equipment (UE) can no longer associate all the transmitters with the same direction of arrival. On the other hand, as wireless systems tend to operate in higher frequency bands, conventional channel models based on rich scattering considerations become less meaningful, and instead line of sight (LoS) channel models are much more accurate and therefore relevant.

In this type of scenario, one can still achieve spatial multiplexing by deploying multiple antennas at both transmitter and receiver (MIMO). In fact, it was shown in \cite{Torkildson2011} that, if both transmitter and receiver are equipped with a uniform linear array (ULA) of $K$ antennas with an inter-element separation $\Delta_T$ and $\Delta_R$ at the transmitter and receiver respectively, one can still achieve $K$ spatial degrees of spatial multiplexing in the LoS channel, provided that the inter-element distances are chosen according to the \textit{Rayleigh spacing} criterion, namely $\Delta_T\Delta_R=\frac{D\lambda}{K}$ where $D$ is the separation between transmit and receive arrays. The benefits from using MIMO-surfaces or MIMO antenna arrays are investigated in \cite{direnzo2023}, and an upper bound on the capacity of LoS MIMO channel is provided in \cite{Do2021} which is attained with uniform linear arrays (ULA) by means of physically rotating the ULA depending on the signal to noise ratio (SNR). 
With the advent of Extremely Large Antenna Arrays (ELAA), also referred to in the literature as eXtra Large (XL)-MIMO, the \textit{far-field} assumption can hardly be met, and the propagation wave front becomes spherical. The \textit{near-field} paradigm offers new opportunities to exploit the wireless medium that have not been considered in conventional wireless system designs. For example, in the \textit{near-field} region it is possible to use spatial filtering to focus energy on a compact region (i.e. \textit{beamfocusing}) rather than steered it into a specific direction \cite{Bjornson19,ramezani2023}. This property  enables large throughput gains in wireless networks by drastically reducing interference among terminals, see \cite{wang2023} to review the latest advances on this front.

\begin{figure}[htbp]
\centerline{\includegraphics[width=3.in, clip=true, trim={0 3 0 3}]{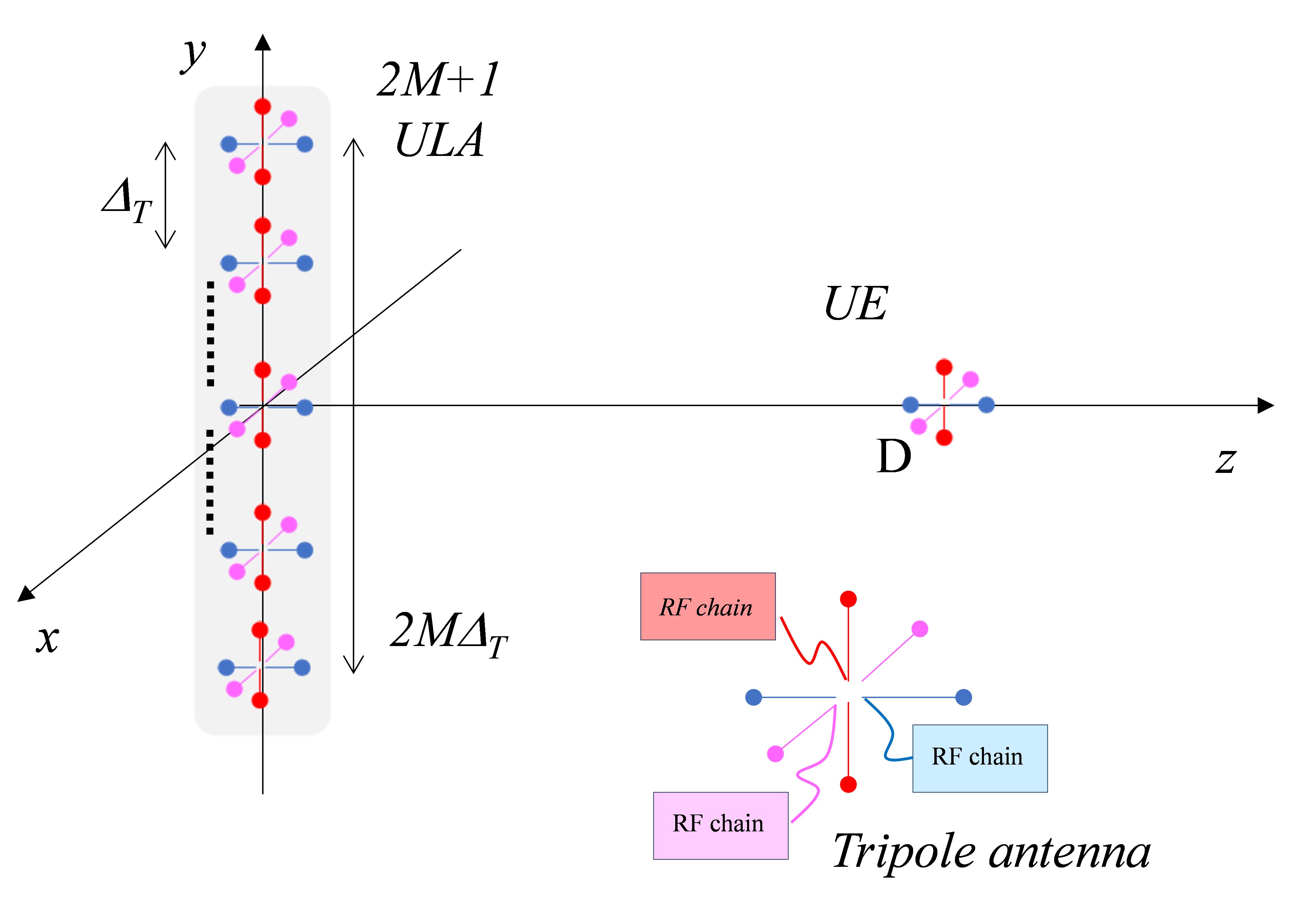}}
\caption{Scenario configuration. Transmitter is based on a ULA with $2M+1$ antenna elements. These elements can be based on: 1 dipole, 2 dipoles or 3 dipoles (in the figure). The UE has a single antenna element.}
\label{fig:Scenario}
\end{figure}

An interesting possibility to increase the spectral efficiency of the ELAA configuration without increasing the total transmission area is to employ multiple polarizations \cite{Chen21,Sena21,Wei23}. The analysis of the electromagnetic propagation phenomena reveals that inter-polarization interference may significantly affect the behavior of the resulting XL-MIMO channel. The attained spectral efficiency in uplink of this type of polarized XL-MIMO architecture is addressed in \cite{Torres2020near}, elucidating that the polarization mismatch can significantly degrade the spectral efficiency.
This paper aims to investigate the use of additional polarizations in a ULA-based XL-MIMO system, see Fig. \ref{fig:Scenario}, when the communication is carried out in the \textit{near-field} region, revisiting the results of \cite{direnzo2023},\cite{Do2021} derived for the \textit{far-field}. Different from \cite{Wei23,Torres2020near}, we will keep the essential characteristics of the electromagnetic propagation conditions, leading us to an increase of the spatial transmission modes as a function of employed polarizations. To simplify the analysis, we will concentrate on the performance of antenna arrays composed of three orthogonal infinitesimal dipoles, for which the total dimension of the antenna is much lower than the operative wavelength. 

The main contributions are as follows. First, we demonstrate that the total spatial degrees of freedom for spatial multiplexing is 3, assuming that the ELAA has at least three elements and the UE is located in the \textit{near-field} receiving with three orthogonal dipoles. We provide an approximation for the achievable rate at high SNR and the optimal antenna spacing, see Lemma \ref{lm:Capacity_highSNR}. Second, we show that, when both UE and ELAA elements consist of 2 orthogonal dipoles only, the maximum achievable rate is obtained when the distance between array elements $\Delta_T$ converges to zero. Consequently, enlarging the ULA will lead to a degradation of the achievable rate for a given UE, but will enhance its focusing capability, i.e. reducing the generated interference to unintended UEs. On the other hand, the use of three orthogonal polarizations will lead to gains in both domains, i.e. in terms of achievable rate and focusing capability. Finally, we reason that when the UE is placed in the \textit{far-field} region, the use of three orthogonal polarizations does provide any benefits compared with the conventional 2-dipole configuration, because the channel matrix tends to be rank-deficient with rank 2, see Lemma \ref{lm:eigenvalues_FF}.

\section{System Model}


We consider a single-user XL-MIMO scenario as it is depicted in in Fig. \ref{fig:Scenario}. The transmitter is a XL-ULA equipped with $2M+1$ elements at $(0,0,0)$, each element consisting of up to 3 orthogonal infinitesimal dipoles. The antenna separation between antennas is denoted by $\Delta_{T}$. A single element UE (which can also consist of up to 3 orthogonal polarizations) is placed at coordinates $(0,0,D)$, where $D$ is the distance between the UE and the center of the array. Notice that each dipole is assumed to be fed by a different radio frequency (RF) unit, allowing for the possibility of multiplexing in the polarization domain. 

Let us denote by $\mathbf{r}_m$ the position vector of the $m$th element of the ULA, $m=-M,\ldots,M$, as seen from the UE, that is 
\begin{IEEEeqnarray}{c}
{\mathbf{r}}_{m} = 
\begin{pmatrix}
0\\
m\Delta_{T}\\
0
\end{pmatrix} -
\begin{pmatrix}
0\\
0\\
D\\
\end{pmatrix}
\label{eq:distance}
\end{IEEEeqnarray}
so that, in particular, $r_m = \|\mathbf{r}_m\|$ is the distance between the $m$th element of the ULA and the UE and $\hat{\mathbf{r}}_{m} = \mathbf{r}_{m}/{r_{m}}$ is the unit norm propagation vector. 

Consider here the transmission from the ULA to the UE. The electric field at the receiver generated by the $m$th transmitting element can be described though the Green function in classical electromagnetic theory, see \cite{Poon05}. The equivalent channel matrix under the assumption of infinitesimal dipole antennas is therefore given by 
\begin{IEEEeqnarray}{c}
\mathbf{H}_{m} = \mathrm{j} \frac{\eta}{2\lambda}\frac{\exp\left({-\mathrm{j}\frac{2\pi}{\lambda }r_{m}}\right)}{r_{m}}\mathbf{P}_{m}\\
\mathbf{P}_{m} = \mathbf{I} - \hat{\mathbf{r}}_{m}\hat{\mathbf{r}}_{m}^H + \mathbf{\Psi}_{m} \\ 
\mathbf{\Psi}_{m} = \frac{\lambda (\mathrm{j} 2\pi r_{m}- \lambda)}{(2\pi r_{m})^2 }\left(\mathbf{I} - 3\hat{\mathbf{r}}_{m}\hat{\mathbf{r}}_{m}^H\right) 
\label{eq:channel_matrix_dipole}
\end{IEEEeqnarray}
where $\eta$ is the permittivity constant of the propagation medium and $\lambda$ is the transmitted wavelength. 
Let us remark that the magnitude of the elements of the matrix $\mathbf{\Psi}_{m}$, which contains the reactive terms of the radiated field, decay to zero as the inverse of the propagation distance. This means that in practice one can safely approximate $\mathbf{\Psi}_{m} \approx \mathbf{0}$ even for short distances, so that matrix $\mathbf{P}_{m}\in \mathbb{C}^{3\times3}$, tends to be an orthogonal projection matrix to the direction defined by $\hat{\mathbf{r}}_{m}$.  

The channel model above can also be particularized to the case where only some of the polarizations are used at either the transmitter or the receiver. More specifically, if the ULA and UE only use the first $t_\mathrm{pol}$ and $r_\mathrm{pol}$ polarizations respectively, the equivalent channel matrix associated with the $m$th transmitter element can be written as $\bar{\mathbf{H}}_m =\mathbf{H}_{m}\left(1:r_\mathrm{pol},1:t_\mathrm{pol}\right) \in \mathbb{C}^{r_\mathrm{pol} \times t_\mathrm{pol}}$. 
Finally, the equivalent channel matrix for the XL-MIMO ULA with $M_t=2M+1$ elements can be expressed as 
$\mathbf{H}_\mathrm{eq} = \left[ \bar{\mathbf{H}}_{-M},\ldots,\bar{\mathbf{H}}_{M}\right] \in \mathbb{C}^{r_\mathrm{pol} \times t_\mathrm{pol} M_t}$. 

Let $\mathbf{x}\in\mathbb{C}^{n_s \times 1}$ denote the transmitted signal at a certain time instant, where $n_s$ is the total number of symbol streams that are spatially multiplexed. Assuming that a certain precoding matrix $\mathbf{F}\in\mathbb{C}^{t_{\mathrm{pol}}M_t\times n_s}$ is applied, the received signal can be expressed as 
\begin{IEEEeqnarray}{c}
\mathbf{y} =\sqrt{\rho}\mathbf{H}_{eq}\mathbf{F}\mathbf{x} + \mathbf{w}
\label{eq:signal_model}
\end{IEEEeqnarray}
where $\mathbf{y} \in \mathbb{C}^{r_\mathrm{pol}\times1}$ denotes the received signal, $\rho$ is the signal to noise ratio and $\mathbf{w} \in \mathbb{C}^{r_\mathrm{pol} \times 1}$ represents the Additive White Gaussian Noise (AWGN) vector, modeled as a standard circularly symmetric normal vector.

\section{Rank of the channel}
\label{sec:rank_channel}
The maximum number of total achievable spatial transmission modes (spatial degrees of freedom) is connected with the rank of the channel matrix, which itself is determined by the number of non-zero singular values. To investigate this point, we consider the eigenvalues of the Gramian of the channel matrix, namely $ \mathbf{H}_{eq}\mathbf{H}_{eq}^H= \mathbf{W} $, where $\mathbf{W} \in \mathbb{C}^{r_\mathrm{pol} \times r_\mathrm{pol}}$ is defined as
\begin{IEEEeqnarray}{c}
\mathbf{W} = \left(\frac{\eta}{2\lambda}\right)^2 
\sum_{m=-M}^M{\!\! \frac{\mathbf{P}_{m}\mathbf{P}_{m}^H}{r_m^2}}.
\label{eq:W_matrix}
\end{IEEEeqnarray}
If transmitter and receiver are equipped with tripole antennas ($t_\mathrm{pol}=r_\mathrm{pol}=3$), matrix $\mathbf{W}$ becomes 
\begin{IEEEeqnarray}{c}
\mathbf{W}^{3\times3} =\left(\frac{\eta}{2\lambda}\right)^2
\begin{pmatrix}
D^2\beta_0\!\! +\!\! \Delta_{T}^2\beta_1\!\! & 0 \!\!&\!\! 0 \\
0  & D^2\beta_0\!\!\!\! &\!\! 0 \\
0\!\! & 0 \!\!& \!\!\Delta_{T}^2\beta_1 \\
\end{pmatrix}
\label{eq:HH_ULA_tripole_tripole}
\end{IEEEeqnarray}
where we have defined
\begin{IEEEeqnarray}{c}
\!\!\beta_0=\!\!\!\!\!\!\sum_{m=\!-M}^{M}{\!\frac{1}{r_{m}^4}}=\!\!\!\!\!\!\sum_{m=\!-M}^{M}{\frac{1}{\left(D^2+  m^2\Delta_{T}^2\right)^2}}
\label{eq:beta_0}\\
\!\!\beta_1=\!\!\!\!\!\!\sum_{m=\!-M}^{M}{\frac{m^2}{r_{m}^4}}=\!\!\!\!\!\!\sum_{m=\!-M}^{M}{\frac{m^2}{\left(D^2 + m^2\Delta_{T}^2\right)^2}}.
\label{eq:beta_1}
\end{IEEEeqnarray}


\begin{lemma}
\label{lm:tripole_tripole}
The eigenvalues of $\mathbf{W}^{3\times3}$ in  (\ref{eq:HH_ULA_tripole_tripole}) can be approximated as a function of $\epsilon\!=\!\frac{M\Delta_T}{D}$ by
\begin{IEEEeqnarray}{c}
\lambda _{1}^{3 \times 3} (\epsilon) = \left(\frac{\eta}{2\lambda}\right)^2\frac{2M}{ D^2}\frac{\arctan\left(\epsilon\right)}{\epsilon} \!\!\!\! \\
\lambda _{2}^{3 \times 3} (\epsilon) = \left(\frac{\eta}{2\lambda}\right)^2\frac{M}{D^2}\left(\!\frac{\arctan\left(\epsilon\right)}{\epsilon} + \frac{1}{\epsilon^2+1}\right)\!\! \\
\lambda _{3}^{3 \times 3} (\epsilon) = \left(\frac{\eta}{2\lambda}\right)^2\frac{M}{D^2}\left(\!\frac{\arctan\left(\epsilon\right)}{\epsilon} - \frac{1}{\epsilon^2+1}\right)
\end{IEEEeqnarray}
where the approximation becomes exact as $M\rightarrow\infty$, $\Delta_T\rightarrow 0$ while $M\Delta_T$ converging to a constant. 
\end{lemma}

\begin{proof}
Assuming that $\Delta_{T} = \frac{D}{\omega}$, and $M$ is large enough, the partial sums introduced in (\ref{eq:beta_0})-(\ref{eq:beta_1}) can be approximated as a definite integral over $(-\!M,M)$, which can be derived as
\begin{IEEEeqnarray}{c}
\!\!\beta_0(\omega) \approx\frac{\omega}{D^4\!\!}\left(\arctan\left(\frac{M}{\omega}\right) + \frac{M\omega}{M^2+\omega^2}\right)\\
\!\!\beta_1(\omega) \approx \frac{\omega^3}{D^4\!}\left(\arctan\left(\frac{M}{\omega}\right) - \frac{M\omega}{M^2+\omega^2}\right).
\end{IEEEeqnarray}
We obtain the eigenvalues with the operations defined in the diagonal elements of (\ref{eq:HH_ULA_tripole_tripole}) and using $\epsilon\!=\!\frac{M}{\omega}\!=\!\frac{M\Delta_T}{D}$.
\end{proof}
When the elements at the transmit ULA consist of 2 orthogonal dipoles (in the $x$ and $y$ directions), but the receive UE implements three dipoles, $(r_\mathrm{pol},t_\mathrm{pol})=(3,2)$ the dimensions of the channel matrix differs from the previous case, so that 
\begin{IEEEeqnarray}{c}
\mathbf{W}^{3\times2}=\left(\frac{\eta}{2\lambda}\right)^2
\begin{pmatrix}
\beta_2\!\!  & 0 \!\!&\!\!0 \\
0  & D^4\beta_3\!\!\!\! &\!\! 0 \\
0 & 0 \!\!&  D^2\Delta_{T}^2\beta_3\\
\end{pmatrix}
\label{eq:HH_ULA_tripole_dipole}
\end{IEEEeqnarray}
where
\begin{IEEEeqnarray}{c}
\beta_2=\!\sum_{m=\!-M}^{M}{\!\!\frac{1}{r_{m}^2}}=\!\!\!\sum_{m=\!-M}^{M}{\frac{1}{\left(D^2+  m^2\Delta_{T}^2\right)}}
\label{eq:beta_2}\\
\!\!\beta_3=\!\!\!\sum_{m=\!-M}^{M}{\frac{1}{r_{m}^6}}=\!\!\!\!\!\!\sum_{m=\!-M}^{M}{\frac{1}{\left(D^2 + m^2\Delta_{T}^2\right)^3}}
\label{eq:beta_3} \\
\!\!\beta_4=\!\!\!\sum_{m=\!-M}^{M}{\frac{m^2}{r_{m}^6}}=\!\!\!\!\!\!\sum_{m=\!-M}^{M}{\frac{m^2}{\left(D^2 + m^2\Delta_{T}^2\right)^3}}
\label{eq:beta_4}
\end{IEEEeqnarray}

\begin{lemma}
\label{lm:tripole_dipole}
Under the same conditions as Lemma \ref{lm:tripole_tripole}, the eigenvalues of $\mathbf{W}^{3\times 2}$ can be approximated by
\begin{IEEEeqnarray}{c}
\lambda _{1}^{3 \times 2} (\epsilon) = \left(\frac{\eta}{2\lambda}\right)^2\frac{2M}{ D^2}\frac{\arctan\left(\epsilon\right)}{\epsilon} \!\!\!\! \\
\lambda _{2}^{3 \times 2} (\epsilon)\!\! = \left(\frac{\eta}{2\lambda}\right)^2\frac{M}{ 4D^2}\left(\frac{3\arctan\left(\epsilon\right)}{\epsilon} +\frac{3\epsilon^2+5}{\left(\epsilon^2+1\right)^2}\right)\!\! \\
\lambda _{3}^{3 \times 2} (\epsilon)\!\! = \left(\frac{\eta}{2\lambda}\right)^2\frac{M}{ 4D^2}\left(\frac{\arctan\left(\epsilon\right)}{\epsilon} +\frac{\epsilon^2-1}{\left(\epsilon^2+1\right)^2}\right)\!\! 
\end{IEEEeqnarray}
\end{lemma}
\begin{proof}
Similar to the \textit{proof} of Lemma \ref{lm:tripole_tripole}. 
\end{proof}

Finally, when all the transceivers use two dipoles oriented as the $x$ and $y$ axis $(t_\mathrm{pol},r_\mathrm{pol})=(2,2)$, the maximum of the rank of the channel is 2.

\begin{lemma} 
\label{lm:dipole_dipole}
Under the same conditions as Lemma \ref{lm:tripole_tripole} above, the eigenvalues of $\mathbf{W}^{2 \times 2}$ can be approximated as
\begin{IEEEeqnarray}{c}
\lambda _{1}^{2 \times 2} (\epsilon) = \left(\frac{\eta}{2\lambda}\right)^2\frac{2M}{ D^2}\frac{\arctan\left(\epsilon\right)}{\epsilon} \!\!\!\! \\
\lambda _{2}^{2 \times 2} (\epsilon)\!\! = \left(\frac{\eta}{2\lambda}\right)^2\frac{M}{ 4D^2}\left(\frac{3\arctan\left(\epsilon\right)}{\epsilon} +\frac{3\epsilon^2+5}{\left(\epsilon^2+1\right)^2}\right)\!\!
\end{IEEEeqnarray}
\end{lemma}

\begin{proof}
It can be shown that Grammian matrix $\mathbf{W}$ in (\ref{eq:W_matrix}) becomes diagonal, having as diagonal elements the same that are presented in position (1,1) and (2,2) in (\ref{eq:HH_ULA_tripole_dipole}). Therefore, the eigenvalues depend on (\ref{eq:beta_2}) and (\ref{eq:beta_3}), and the result follows from the 
definition of Riemann integral.
\end{proof}

\section{Optimal ULA size}
This section establishes the fact that, for some polarization combinations, there exists an optimum size of the transmit ULA in terms of maximum achievable rate. Indeed, when the transmitter has perfect channel state information, 
the achievable rate can be described as
$$
C(\rho)=\sum_{n=1}^{3} \mathrm{log}_2\left( 1+ {\rho}\! \left[ \frac{1}{\gamma}\! - \!\frac{1}{\lambda_n} \right]^{+}\lambda_n \right)
$$
where $\gamma$ denotes the power allocation water-level, $[z]^+=\mathrm{min}(0,z)$, $\rho$ is the signal to noise ratio and $\lambda_n$ are the channel eigenvalues studied in Section \ref{sec:rank_channel} for each polarization configuration. The number of spatial degrees of freedom (DoF), or maximum number of spatial transmission modes with non-zero power allocation, is defined by 
\begin{IEEEeqnarray}{c}
\nu = \lim_{\rho\rightarrow\infty}\frac{\log 2^{C(\rho)}}{\log(\rho)} 
\label{eq:dof}
\end{IEEEeqnarray}

\begin{lemma}
\label{lm:Capacity_highSNR}
In a scenario with a ULA with $M_t=2M\!+\!1$ elements using polarizations $t_\mathrm{pol}\in\{2,3\}$, and one UE with polarization $r_\mathrm{pol}\in\{2,3\}$, the maximum achievable rate at high SNR can be approximated (in the sense of Lemma \ref{lm:tripole_tripole}) as
\begin{IEEEeqnarray}{c}
C=C_0(\rho,M,D) + \alpha\left(\frac{M\Delta_T^*}{D}\right)
\label{eq:Capacity_highSNR}
\end{IEEEeqnarray}
where
\begin{IEEEeqnarray}{c}
C_0(\rho,M,D)=
\begin{cases}
  3\log_2\left( \frac{\eta}{2\lambda}\right)^2\frac{M}{D^2}\frac{\rho}{3} &  (r_\mathrm{pol},t_\mathrm{pol})=(3\times3) \\
  3\log_2\left( \frac{\eta}{2\lambda}\right)^2\frac{M}{4D^2}\frac{\rho}{3}  &  (r_\mathrm{pol},t_\mathrm{pol})=(3\times2)\\
  2\log_2\left( \frac{\eta}{2\lambda}\right)^2\frac{M}{4D^2}\frac{\rho}{2} &  (r_\mathrm{pol},t_\mathrm{pol})=(2\times2)
  \end{cases} \nonumber
\end{IEEEeqnarray}
and
\begin{IEEEeqnarray}{c}
\alpha \left(\frac{M\Delta_T^*}{D}\right) =
\begin{cases}
  -0.7794 &  (r_\mathrm{pol},t_\mathrm{pol})=(3\times3) \\
  4.6339  &  (r_\mathrm{pol},t_\mathrm{pol})=(3\times2)\\
  \approx 0 &  (r_\mathrm{pol},t_\mathrm{pol})=(2\times2).
  \end{cases} \nonumber
\end{IEEEeqnarray}
This is achieved with the optimal $\Delta_T$, given by
\begin{IEEEeqnarray}{c}
\label{eq:op_sep_antenna}
\Delta_T^* =
    \begin{cases}
      0.9058\frac{D}{M} &  (r_\mathrm{pol},t_\mathrm{pol})=(3\times3) \\
      0.7144\frac{D}{M} & (r_\mathrm{pol},t_\mathrm{pol})=(3\times2)\\
       \approx 0 & (r_\mathrm{pol},t_\mathrm{pol})=(2\times2).
    \end{cases}     
\end{IEEEeqnarray}

\end{lemma}
\begin{proof}
By using the expression of the approximated eigenvalues one readily sees that the achievable rate at high SNR follows the expression in (\ref{eq:Capacity_highSNR}). The second term, $\alpha$, depends on the polarization configuration $(r_\mathrm{pol},t_\mathrm{pol})$ and $\epsilon=\frac{M\Delta_T}{D}$ as 
\begin{IEEEeqnarray}{c}
\alpha^{3\times3}(\epsilon)\! = \!\log_2\left(\!2\zeta\left(\!\!\zeta+\frac{1}{\epsilon^2+1}\!\!\right)\!\left(\!\!\zeta-\frac{1}{\epsilon^2+1}\!\!\right)\!\!\right)\!\!
\label{eq:alpha_3x3}
\end{IEEEeqnarray}
\begin{IEEEeqnarray}{c}
\!\!\!\alpha^{3\times2}(\epsilon)\! = \!\log_2\left(\!{8}\zeta\left(\!\!3\zeta+\frac{3\epsilon^2+5}{(\epsilon^2+1)^2}\!\!\right)\!\left(\!\!\zeta+\frac{\epsilon^2-1}{(\epsilon^2+1)^2}\!\!\right)\!\!\right)\!\!
\label{eq:alpha_3x2}
\end{IEEEeqnarray}
\begin{IEEEeqnarray}{c}
\alpha^{2\times2}(\epsilon)\! = \! \log_2\left(\!{8}\zeta\left(\!\!3\zeta+\frac{3\epsilon^2+5}{(\epsilon^2+1)^2}\!\!\right)\!\!\right)\!\!
\label{eq:alpha_2x2}
\end{IEEEeqnarray}
where $ \zeta \!\!= \!\!\frac{\arctan\left(\epsilon\right)}{\epsilon}$. Since the maximum only depends on variable $\epsilon$, the optimization can be done numerically. \end{proof}
Fig. \ref{fig:optimal_alpha} presents $\alpha^{r_\mathrm{pol}\times t_\mathrm{pol}}$ as a function of $\epsilon$. When $(r_\mathrm{pol},t_\mathrm{pol})\in \{(3,3), (3,2)\}$, $\alpha$ has a maximum for certain $\epsilon$, while for $(r_\mathrm{pol},t_\mathrm{pol})=(2,2)$, $\alpha^{2\times2}$ is a monotonically decreasing function, implying that the achievable rate under configuration $(r_\mathrm{pol},t_\mathrm{pol})\!=\!(2,2)$ always decreases if the antenna separation increases.
\begin{figure}[htbp]
\centerline{\includegraphics[width=3.1in]{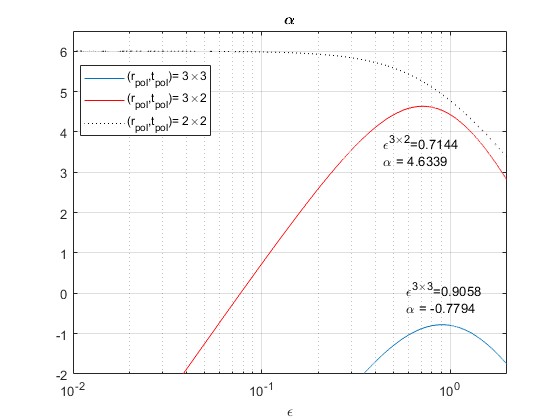}}
\caption{ Evaluation of function $\alpha$ presented in (\ref{eq:alpha_3x3}),(\ref{eq:alpha_3x2}),(\ref{eq:alpha_2x2})}
\label{fig:optimal_alpha}
\end{figure}
Lemma \ref{lm:Capacity_highSNR} highlights the importance of using three orthogonal dipoles to increase the number of DoF  in the \textit{near-field} regime ($\nu^{3\times3}=\nu^{3\times2}=3$) with respect the case of having just two orthogonal dipole antennas ($\nu^{2\times2}=2$). The following lemma shows that, when the UE is far away from the antenna array and the inter-element separation is fixed (in general, $\Delta_T \neq \Delta_T^*$), the use of three orthogonal polarization does not bring additional spatial degrees of freedom. The proof follows from direct evaluation of the expressions given in Lemmas \ref{lm:tripole_tripole},\ref{lm:tripole_dipole} and \ref{lm:dipole_dipole} after assuming assuming $D \gg \Delta_TM$ and $D \gg 0$. 
\begin{lemma}
\label{lm:eigenvalues_FF}
As $D \rightarrow \infty$ the the rank of the channel matrix with dimensions $r_\mathrm{pol}\times t_\mathrm{pol}M_t$ converges to 2, so that $\nu^{3\times3},\nu^{3\times2}$ and $\nu^{2 \times 2}$ also converge to $2$.
\end{lemma}

%
%

\section{Numerical validation}
We consider the scenario presented in Fig. \ref{fig:Scenario} with a carrier frequency of $f=3$ GHz ($\lambda=0.1$) and $\eta=1$. As before, we consider three different configurations regarding polarizations, which will be denoted as $(r_\mathrm{pol} \!\times\! t_\mathrm{pol})=(3,3)$, $(r_\mathrm{pol} \!\times\! t_\mathrm{pol})=(3,2)$ and $(r_\mathrm{pol} \!\times\! t_\mathrm{pol})=(2,2)$ respectively. 

\subsection{Eigenvalues of the channel matrix}\label{sec:res_eigenvalues}
Fig.~\ref{fig:eigenvalues_plot} illustrates the eigenvalues of matrix $\mathbf{H}_{eq}\mathbf{H}_{eq}^H$ for the different polarization configurations. We compare the results obtained by simulation (solid lines) and the ones obtained with the expressions derived in Lemma \ref{lm:tripole_tripole}, Lemma  \ref{lm:tripole_dipole}  and Lemma \ref{lm:dipole_dipole} (dotted lines) when the UE is at $D=5$ and there are $M_t \in \{7,31\}$ antenna elements. Results are obtained as a function on the antenna separation $\Delta_T$. We verify that the analytical solutions tends to the simulated one when we increase $M$, which agrees with the asymptotic regime considered to derive these expressions. Additionally, as the analytical expressions suggest, we see that: (i) the first eigenvalue of the channel for all the configurations is the same ($\lambda_1^{3\times3}=\lambda_1^{3\times2}=\lambda_1^{2\times2}$); (ii) the second eigenvalues of the configurations $(2,2)$ and $(2,3)$ are equal ($\lambda_2^{3\times2}=\lambda_2^{2\times2}$) but smaller than the first eigenvalue with configuration $(3,3)$ ($\lambda_2^{3\times3} \geq \lambda_2^{3\times2}$); and (iii) both the first and second channel eigenvalues are decreasing functions of $\Delta_T$, although the third non-zero eigenvalue in configurations $\{(3,3)$ and $(3,2)\}$ presents a non-trivial maximum with respect to $\Delta_T$. 
\begin{figure}[htbp]
\centerline{\includegraphics[width=3.65in]{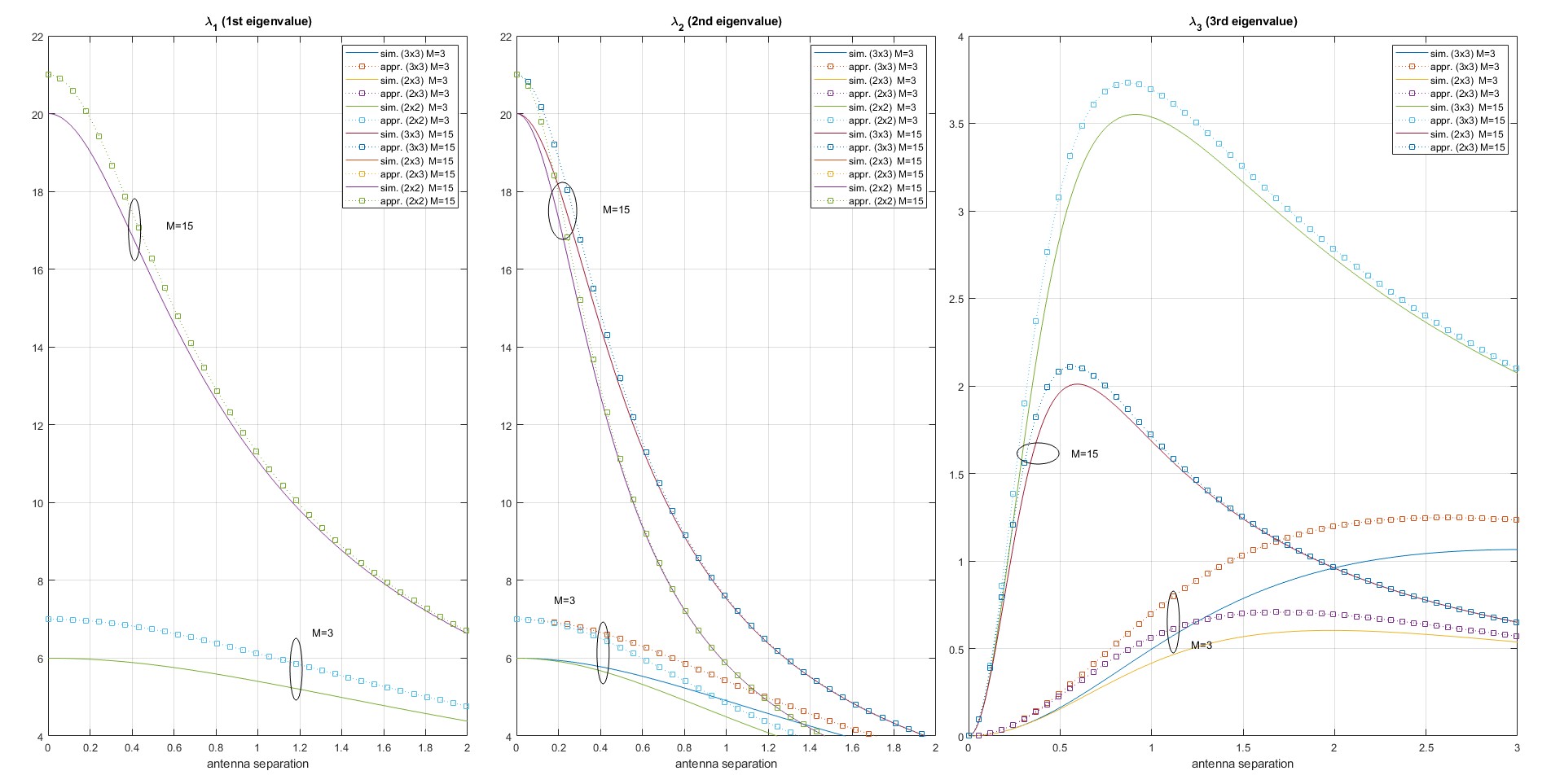}}
\caption{1st eigenvalue (left), 2nd eigenvalue (center) and 3rd eigenvalue (right) obtained by simulation and analytically when the UE is at $D=5$ and $M_t \in {(7,31)}$ ($M \in {(3,15)}$).}
\label{fig:eigenvalues_plot}
\end{figure}

\subsection{Achievable rate}\label{sec:res_capacity}
Fig. \ref{fig:Capacity_vs_SNR} presents the maximum achievable rate for a given SNR ($\rho$) under multiple polarization configurations when there are $M_t=41$ antenna elements and UE is at $D=5$. The maximum achievable rate is obtained by optimizing the antenna spacing $\Delta_T$. Configurations $(3,3),(2,2)$ show the largest gains in terms of achievable rate. Furthermore, at high SNR, the slope of the achievable rate, measured according to (\ref{eq:dof}), attains $\nu\!=\!3$. Likewise, we can observe that the approximations provided in Lemma \ref{lm:Capacity_highSNR} for the achievable rate at high SNR are very close to the simulated achievable rate when $\rho\! \geq\! 5$ dB. On the other hand, Fig. \ref{fig:ULAsize} depicts the ULA size where the achievable rate is maximum as a function of the SNR and the UE is at $D=\{3,5\}$. Notice that the optimal ULA sizes coincide with the values predicted by Lemma \ref{lm:Capacity_highSNR}, i.e. $L_{\text{ULA}}^{3\times3}\!=\!2M\Delta_T^*\!\!=\!\!2\! \times\! 0.9058 D$ and $L_{\text{ULA}}^{3\times2}\!\!=\!\!2 \!\times\! 0.7144 D$. Interestingly enough, the obtained achievable rate with configuration $(3,2)$ is close to the maximum one $(3,3)$, but requiring an smaller ULA size with only three orthogonal dipoles at the UE side.

\begin{figure}[htbp]
\centerline{\includegraphics[width=3.2in]{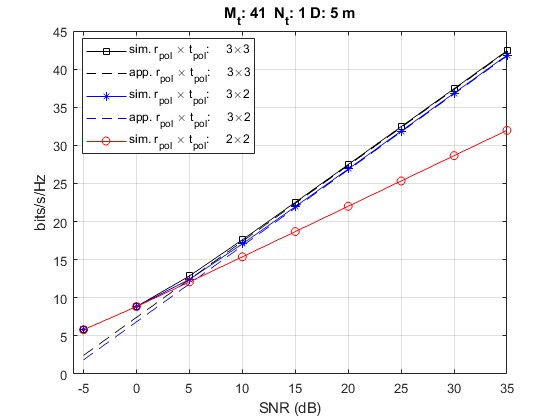  }}
\caption{Optimized achievable rate (simulated and using approximation of Lemma \ref{lm:Capacity_highSNR}) as a function of SNR (dB) with respect antenna separation ($\Delta_T$). UE at $D=5$. $M_t=41$. }
\label{fig:Capacity_vs_SNR}
\end{figure}

\begin{figure}[htbp]
\centerline{\includegraphics[width=3.2in]{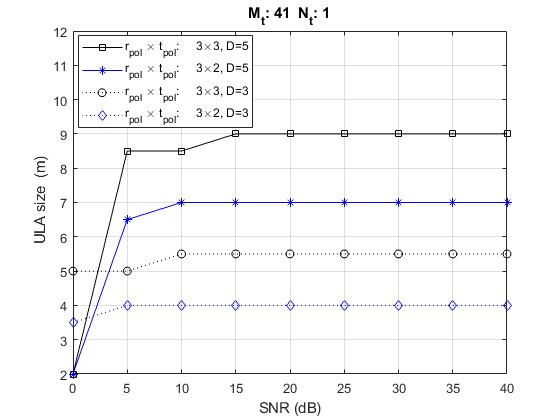}}
\caption{Optimal size of the ULA (numerical optimization of $\Delta_T^*$) for a UE with SNR $\rho=20$ dB, $M_t=41$. The predicted size is given by Lemma \ref{lm:Capacity_highSNR}: $L_{\text{ULA}}^{3\times3}=2\! \times\! 0.9058 D$ and $L_{\text{ULA}}^{3\times2}=2 \!\times\! 0.7144 D$. }
\label{fig:ULAsize}
\end{figure}



\subsection{Near-field beam-focusing} \label{res:sec_beamfocusing}
In this section, we would like to study the \textit{near-field} beamfocusing capabilities of the ULA, which can be measured as the effect of a precoding matrix $\mathbf{F}$, which is originally designed for a certain UE location, on the signal that could potentially be received in other spatial positions. 
In order to evaluate this effect, we assume that we design the transmit filter $\mathbf{F}$ for a terminal placed at $D=5$, and we evaluate the achievable rate that could be obtained by a terminal placed in positions along the \textit{z}-axis $(0,0,D)$. If the polarized ULA is able to focus the transmission, we would expect that a UE, for example at $D=10$, would not be interfered by the signal intended to the UE at $D=5$, which means that it would measure a very low achievable rate. In this regard, Fig.~\ref{fig:focusing_2x2}-(top) represents the achievable rate obtained when having  $2M+1$ antenna elements at the ULA with a polarization configuration $(2\times2)$, total antenna size of $L_{\textit{ULA}}^{2\times2}=0.5$ m, and a transmit precoder designed for a UE at $D=5$. For this configuration, we could increase the antenna size, see Fig.~\ref{fig:focusing_2x2}-(bottom), and results show that the signal can be slightly focused at the position of the UE. This is a consequence of having a large ULA, operating in the \textit{near-field} region in combination with a precoder that takes into account the spherical wave propagation pheomenon. Observe, in any case, that this comes at the cost of reducing the maximum achievable rate at the point of interest. Recall from Lemma \ref{lm:Capacity_highSNR} that for the $(2\times2)$ polarization configuration the optimum rate is obtained as $\Delta_T \rightarrow 0$.

\begin{figure}[htbp]
\centerline{\includegraphics[trim=0 15 0 15, clip=1,width=3.6in]{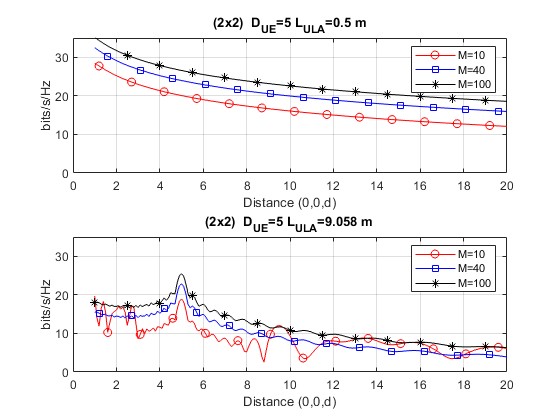}}
\caption{Evaluation of the achievable rate as a function of the distance $(0,0,d)$ when transmit precoders are selected for a UE placed at $(0,0,5)$. SNR $\rho=20$ dB. Top: configuration $(2\times2)$ with $L_{\text{ULA}}$=0.5 m,  Bottom: configuration $(2\times2)$ with $L_{\text{ULA}}$=9.058 m}
\label{fig:focusing_2x2}
\end{figure}

On the other hand, Fig.~\ref{fig:focusing_3x3_3x2} represents the achievable rate obtained when three orthogonal dipoles are considered, which allows increasing the number of spatial transmission modes from $\nu=2$ to $\nu=3$. Configurations $(3\times3)$ and $(2\times2)$ are considered in Fig.~\ref{fig:focusing_3x3_3x2}-(top) and Fig.~\ref{fig:focusing_3x3_3x2}-(bottom), respectively. In both cases, we have set the antenna spacing according to Lemma \ref{lm:Capacity_highSNR}, leading to different ULA sizes. It can be observed that the maximum achievable rate of each configuration is very similar, but the region where the signal is focused is slightly different. For example with $M=100$, the region where the achievable rate is at 10 bits/s/Hz below the maximum in configuration $(3\times3)$ is in the region $d \in (4.6, 5.35)$, while for the $(3\times2)$ configuration is $d\in (4.55, 5.45)$.

\begin{figure}[htbp]
\centerline{\includegraphics[trim=0 15 0 15, clip=1,width=3.6in]{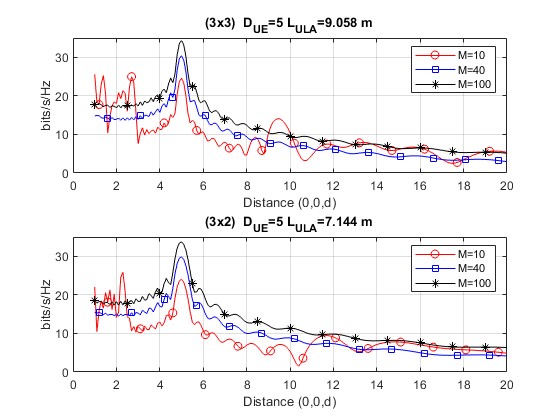}}
\caption{Evaluation of the achievable rate as a function of the distance $(0,0,d)$ when transmit precoders are selected for a UE placed at $(0,0,5)$. SNR $\rho=20$ dB. Top: configuration $(3\times3)$ with $L_{\text{ULA}}=9.058$ m,  Bottom: configuration $(3 \times 2)$ with $L_{\text{ULA}}$=7.144 m.}
\label{fig:focusing_3x3_3x2}
\end{figure}

We have optimized the ULA size in order to maximize the achievable rate at a given distance and SNR, and in fact the transmission can be focused on a certain region, but still the unintended UEs at $d \neq D_{\text{UE}}$ are receiving a important interference component. Making good use of the \textit{peak} in terms of achievable rate, we could reduce the transmitting power to reduce the interference signal at the cost of reducing the achievable rate at the UE's position. Fig. \ref{fig:focusing_all_SNR10} depicts the achievable rate at SNR, $\rho=10$ dB, illustrating that the signal is more concentrated around the UE's position, while reducing the generated interference (i.e. low achievable rate at positions $d\neq D$). The configuration with three orthogonal polarizations is able to improve the achievable rate using 2-dipole antennas by a factor $27\%$.

\begin{figure}[htbp]
\centerline{\includegraphics[width=3.5in]{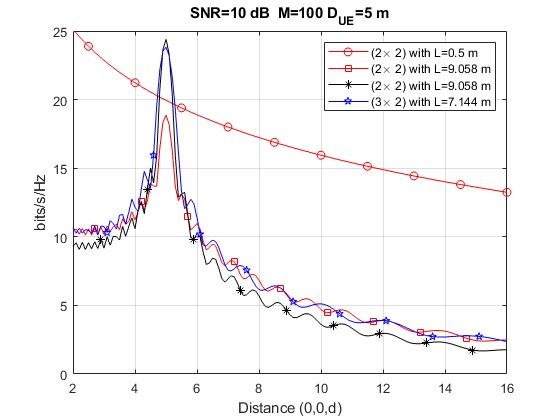}}
\caption{Evaluation of the achievable rate as a function of the distance $(0,0,d)$ when precoder is designed for a UE placed at $D\!=\!5$. SNR $\rho\!=\!10$ dB.}
\label{fig:focusing_all_SNR10}
\end{figure}

\section{Conclusions}
This work has investigated the circumstances under which a transmitter consisting of a single ULA equipped with $2M+1$ antennas is able to focus the signal into a given region. Each antenna element can have up to three orthogonal dipoles. When using two polarizations at both transmitter and receiver sides, we can focus the transmitted signal by increasing the ULA size at the cost of reducing the maximum achievable rate. However, when three orthogonal polarizations are used at least at the receiver side, it is possible to increase the number of spatial transmission modes up to $\nu=3$ and the achievable rate (see Lemma \ref{lm:Capacity_highSNR}) while still spatially focus the transmission. In this case the ULA size tends to $L^{3\times 2}=2\times0.9058\times D$ and $L^{3\times 2}=2\times0.7144\times D$ if the antenna elements at the ULA are respectively using three and two polarizations. 


\bibliographystyle{IEEEtran}
\bibliography{./biblio}

\begin{thebibliography}{10}
\providecommand{\url}[1]{#1}
\csname url@samestyle\endcsname
\providecommand{\newblock}{\relax}
\providecommand{\bibinfo}[2]{#2}
\providecommand{\BIBentrySTDinterwordspacing}{\spaceskip=0pt\relax}
\providecommand{\BIBentryALTinterwordstretchfactor}{4}
\providecommand{\BIBentryALTinterwordspacing}{\spaceskip=\fontdimen2\font plus
\BIBentryALTinterwordstretchfactor\fontdimen3\font minus \fontdimen4\font\relax}
\providecommand{\BIBforeignlanguage}[2]{{%
\expandafter\ifx\csname l@#1\endcsname\relax
\typeout{** WARNING: IEEEtran.bst: No hyphenation pattern has been}%
\typeout{** loaded for the language `#1'. Using the pattern for}%
\typeout{** the default language instead.}%
\else
\language=\csname l@#1\endcsname
\fi
#2}}
\providecommand{\BIBdecl}{\relax}
\BIBdecl

\bibitem{Torkildson2011}
E.~Torkildson, U.~Madhow, and M.~Rodwell, ``Indoor millimeter wave {MIMO}: Feasibility and performance,'' \emph{IEEE Transactions on Wireless Communications}, vol.~10, no.~12, pp. 4150--4160, 2011.

\bibitem{direnzo2023}
M.~D. Renzo, D.~Dardari, and N.~Decarli, ``{LoS} {MIMO}-arrays vs. {LoS} {MIMO}-surfaces,'' 2023.

\bibitem{Do2021}
H.~Do, N.~Lee, and A.~Lozano, ``Reconfigurable {ULAs} for line-of-sight {MIMO} transmission,'' \emph{IEEE Transactions on Wireless Communications}, vol.~20, no.~5, pp. 2933--2947, 2021.

\bibitem{Bjornson19}
E.~Björnson, L.~Sanguinetti, H.~Wymeersch, J.~Hoydis, and T.~L. Marzetta, ``Massive {MIMO} is a reality—what is next?: Five promising research directions for antenna arrays,'' \emph{Digital Signal Processing}, vol.~94, pp. 3--20, 2019.

\bibitem{ramezani2023}
P.~Ramezani, A.~Kosasih, A.~Irshad, and E.~Björnson, ``Massive spatial multiplexing: Vision, foundations, and challenges,'' 2023.

\bibitem{wang2023}
Z.~Wang, J.~Zhang, H.~Du, D.~Niyato, S.~Cui, B.~Ai, M.~Debbah, K.~B. Letaief, and H.~V. Poor, ``A tutorial on extremely large-scale mimo for 6g: Fundamentals, signal processing, and applications,'' 2023.

\bibitem{Chen21}
X.~Chen, J.~C. Ke, W.~Tang, M.~Z. Chen, J.~Y. Dai, E.~Basar, S.~Jin, Q.~Cheng, and T.~J. Cui, ``Design and implementation of {MIMO} transmission based on dual-polarized reconfigurable intelligent surface,'' \emph{IEEE Wireless Comm. Letters}, vol.~10, no.~10, pp. 2155--2159, 2021.

\bibitem{Sena21}
A.~S. de~Sena, P.~H.~J. Nardelli, D.~B.~d. Costa, U.~S. Dias, P.~Popovski, and C.~B. Papadias, ``Dual-polarized {IRSs} in uplink {MIMO-NOMA} networks: An interference mitigation approach,'' \emph{IEEE Wireless Communications Letters}, vol.~10, no.~10, pp. 2284--2288, 2021.

\bibitem{Wei23}
L.~Wei, C.~Huang, G.~C. Alexandropoulos, Z.~Yang, J.~Yang, W.~E.~I. Sha, Z.~Zhang, M.~Debbah, and C.~Yuen, ``Tri-polarized holographic {MIMO} surfaces for near-field communications: Channel modeling and precoding design,'' \emph{IEEE Trans. on Wireless Comms.}, pp. 1--1, 2023.

\bibitem{Torres2020near}
A.~de~Jesus~Torres, L.~Sanguinetti, and E.~Björnson, ``Near- and far-field communications with large intelligent surfaces,'' in \emph{2020 54th Asilomar Conference on Signals, Systems, and Computers}, 2020, pp. 564--568.

\bibitem{Poon05}
A.~Poon, R.~Brodersen, and D.~Tse, ``Degrees of freedom in multiple-antenna channels: a signal space approach,'' \emph{IEEE Transactions on Information Theory}, vol.~51, no.~2, pp. 523--536, 2005.

\end{thebibliography}



\end{document}